\documentclass[journal]{IEEEtran}
\usepackage{graphics} 
\usepackage{epsfig} 
\usepackage{times} 
\usepackage{amsmath} 
\usepackage{amsthm}
\usepackage{amssymb}
\usepackage[utf8]{inputenc}
\usepackage{tikz}\usepackage{enumitem}
\usetikzlibrary{positioning, arrows.meta, automata}
\usetikzlibrary{graphs}
\usepackage[colorlinks=false, urlcolor=blue, pdfborder={0 0 0}]{hyperref}
\usepackage{enumerate}
\usepackage{color}
\usepackage[linesnumbered,vlined,ruled]{algorithm2e}
\usepackage{mathrsfs}
\usepackage[most]{tcolorbox}
\tcbuselibrary{breakable,skins}
\usepackage{subfigure}
\usepackage{comment}
\usepackage{cite}
\usepackage{flushend}  
\usepackage{tikz}
\usetikzlibrary{arrows.meta,positioning,fit,backgrounds,calc}

\newtheorem{mythm}{Theorem}
\newtheorem{myprob}{Problem}

\newtheorem{remark}{Remark}

\def \L{{\mathcal{L}}}

\usepackage{tikz} 

\usepackage{url}
\usepackage{pgf}
\usetikzlibrary{positioning, arrows.meta,automata}

\usetikzlibrary{positioning,shapes}
\usetikzlibrary{arrows}

\tcbuselibrary{skins}
\newtcolorbox{resp}[1][]{%
enhanced jigsaw,%
colback=gray!5!white,%
colframe=gray!80!black,%
size=small,%
boxrule=1pt,%
halign title=flush center,%
coltitle=black,%
breakable,%
drop shadow=black!50!white,%
attach boxed title to top left={xshift=1cm,yshift=-\tcboxedtitleheight/2,yshifttext=-\tcboxedtitleheight/2},%
minipage boxed title=3cm,%
boxed title style={%
	colback=white,%
	size=fbox,%
	boxrule=1pt,%
	boxsep=2pt,%
	underlay={%
		\coordinate (dotA) at ($(interior.west) + (-0.5pt,0)$);
		\coordinate (dotB) at ($(interior.east) + (0.5pt,0)$);
		\begin{scope}[gray!80!black]
			\fill (dotA) circle (2pt);
			\fill (dotB) circle (2pt);
		\end{scope}
	}%
},%
#1%
}

\title{Decentralized Decision-Making  for\\ Finite-State Systems over Finite  Alphabets is Undecidable}

\author{Xiang Yin 

\thanks{This work was supported by  the National Natural Science Foundation of China (62173226, 62061136004, 92367203). }
\thanks{X. Yin is with the School of Automation and Intelligent Sensing, Shanghai Jiao Tong University, Shanghai 200240, China, and also with the Key Laboratory of System Control and Information Processing, the Ministry of Education of China, Shanghai 200240, China. {\tt  E-mail: \{yinxiang\}@sjtu.edu.cn}.}
}

\begin{document}

\maketitle

\begin{abstract}
This paper investigates decentralized decision-making for finite-state transition systems, i.e., discrete-event systems, under \emph{finite communication alphabet} constraints. We consider a general decentralized observation framework in which a plant is observed by multiple local agents that transmit symbolic messages over a finite alphabet to a memoryless fusion center. The fusion center then produces a binary decision according to a prescribed fusion rule.
We study the fundamental question of whether there exist local decision maps that enable exact reconstruction of a given regular specification language from decentralized observations. Contrary to classical results that rely on specific monotone fusion rules such as conjunction and disjunction, we show that the problem becomes undecidable even under a severely restricted information architecture: binary local decision alphabets and a fixed exclusive-or (XOR) fusion rule. The proof is based on a reduction from the Thue word problem, a classical undecidable problem in rewriting systems. We further show that decentralized supervisory control, decentralized fault diagnosis, and decentralized fault prognosis are also undecidable under finite communication alphabets.
Our results reveal that existing decidability results fundamentally rely on structural properties of fusion rules, in particular their monotone order-preserving nature. In contrast, non-monotone fusion rules such as XOR break this structure, leading to undecidability even in highly restricted settings.
\end{abstract}


\IEEEpeerreviewmaketitle

\section{Introduction}

\subsection{Background}
Decentralized decision-making is a fundamental problem in systems and control theory, with both theoretical significance and practical relevance \cite{witsenhausen1971separation,tsitsiklis1985complexity,nayyar2013decentralized}. From a practical perspective, many large-scale systems are inherently decentralized, and communication among local sites may be costly  or even infeasible. From a theoretical perspective, decentralized decision-making involves an information structure that is fundamentally different from its centralized counterpart.
In particular, due to the lack of communication among agents, each local agent must reason not only about the plant state, but also about the knowledge available to other agents. Such knowledge, in turn, may depend on other agents' beliefs about its own knowledge, leading to nested information dependencies. This recursive structure often gives rise to an infinite information state space, which constitutes a   fundamental technical challenge  in decentralized decision-making.

In this work, we investigate decentralized decision-making problems for finite-state transition systems, namely discrete-event systems (DES) \cite{wonham2019supervisory,cassandras2021introduction}.
Such systems provide finite abstractions of more complex dynamical systems due to their finite-state nature, while preserving the essential structure of decentralized decision-making. Consequently, DES offer a simple yet expressive framework for studying decentralized decision-making problems, which  have been extensively explored since the seminal work in the late 1980s in the context of decentralized supervisory control \cite{cieslak1988supervisory,lin1988decentralized,lin1990decentralized,rudie1992think,overkamp2000maximal,ricker2007knowledge,komenda2017computation,tripakis2021decentralized,ritsuka2023you,komenda2023supervisory,ritsuka2024uniform}. Subsequently, various extensions have been developed, including decentralized fault diagnosis and decentralized fault prognosis.

\begin{figure}[tp]
	\centering
	\includegraphics[width=0.9\linewidth]{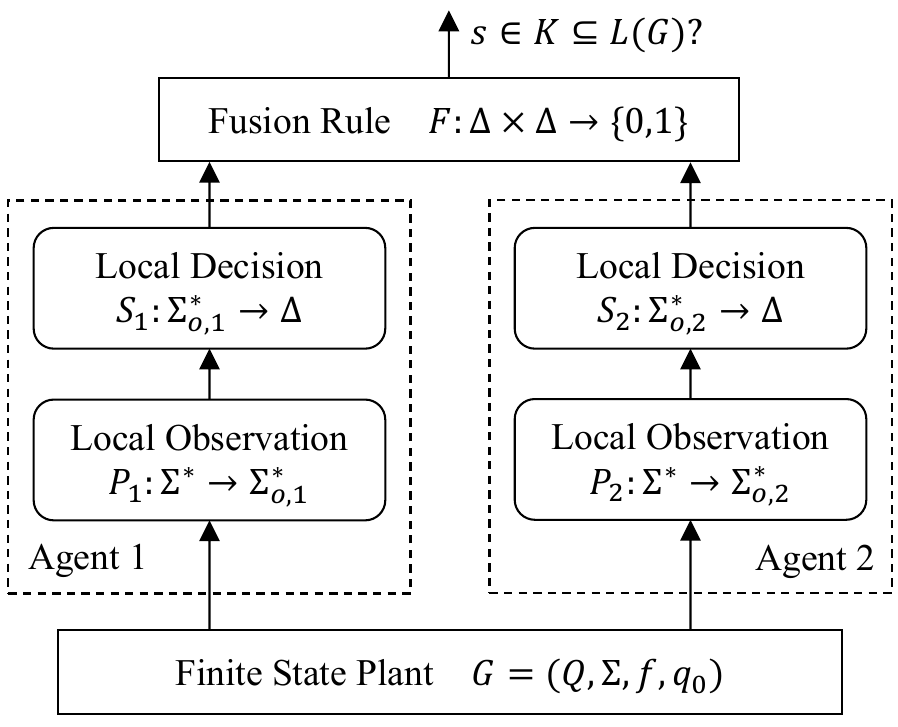}
	\caption{Conceptual illustration of the decentralized decision-making  architecture, where local agents cannot communicate with each other and can only transmit their decisions to a fusion center, which produces the final decision.}
	\label{fig:architecture}
\end{figure}

Despite differences in application contexts, the core structure of these decentralized decision-making problems can be unified under the framework of a decentralized observation problem, as illustrated in Fig.~\ref{fig:architecture}. Specifically, let $G$ denote the finite-state plant, which is observed by a set of local agents (without loss of generality, we consider two agents in this work). Each agent observes the plant through its own locally observable event set $\Sigma_{o,i}$ and generates a local decision symbol from a finite  alphabet $\Delta$. These local decisions are then transmitted to a central fusion center.
The fusion center applies a memoryless fusion rule $F:\Delta \times \Delta \to \{0,1\}$,  
which produces the final binary decision. Depending on the application, this decision may be interpreted as indicating whether an event should be disabled or whether a fault has occurred. In this broader view, the problem reduces to determining whether the observed execution belongs to a given sublanguage of the plant language. The tuple $(\Delta, F)$ is referred to as the \textbf{decentralized architecture} of the  system, which characterizes \emph{what information each local agent is allowed to transmit to the fusion center}, as well as \emph{how the fusion center processes this information to produce the final decision}.
Given a decentralized architecture, the problem is to determine whether there exist local decision maps $S_i:\Sigma_{o,i}^* \to \Delta$ such that the global specification can be enforced by the fusion center.

\subsection{Existing Results and Open Challenges}
In the past three decades, numerous works have been developed for decentralized decision-making under different prespecified information architectures. In this subsection, we focus primarily on discussing results from the decentralized supervisory control literature, since control problems are typically the first to be studied under any newly proposed information architecture; corresponding results in fault diagnosis and fault prognosis are analogous \cite{debouk2000coordinated,wang2007diagnosis,kumar2009decentralized,takai2010inference,moreira2011polynomial,yin2018decentralized,takai2016generalized,takai2026intersection}.
The first and most extensively studied architecture in decentralized supervisory control is the \emph{conjunctive architecture} \cite{cieslak1988supervisory,rudie1992think}, in which each supervisor makes a binary decision, namely ``enable" or ``disable", and the fusion rule is ``disable if at least one supervisor disables''. Subsequently, the \emph{disjunctive architecture} was introduced \cite{prosser1997decision}, where the fusion rule is ``enable if at least one supervisor enables''. In \cite{yoo2002general}, a more general architecture was proposed by allowing different fusion rules for different controllable events. These architectures are commonly referred to as \emph{unconditional architectures}, since the decision alphabet is binary, i.e., $|\Delta|=2$.
For these architectures, necessary and sufficient conditions, typically formulated in terms of variants of coobservability, have been developed for the existence of local decision maps (i.e., supervisors). These conditions are polynomial-time verifiable in the size of the plant, although they are PSPACE-hard in the number of agents \cite{rudie1995computational,rohloff2003deciding}.

In \cite{yoo2004decentralized}, a \emph{conditional architecture} was proposed, in which supervisors are allowed to make conditional decisions such as ``disable if nobody enables'' and ``enable if nobody disables''. The corresponding necessary and sufficient condition is referred to as \emph{conditional coobservability}.
Later, this framework was further generalized to the \emph{inference-based architecture} \cite{takai2008synthesis,takai2024nonexistence}, where multi-level inference among local agents is allowed, and the corresponding notion is called $N$-inference-based observability, where $N$ denotes the depth of inference.
More recently, an \emph{intersection-based architecture} was proposed \cite{yin2016decentralized,hayano2023general,ritsuka2025universal}, in which the local decision is the state estimate from each agent's perspective, i.e., $\Delta = 2^Q$, where $Q$ is the state space. This architecture has been shown to be incomparable with the inference-based architecture.
In essence, these approaches can be interpreted as progressively enlarging the local decision alphabet $\Delta$ while allowing increasingly sophisticated fusion rules.

In all the aforementioned architectures, the verification of the existence of local decision maps relies heavily on natural control-theoretic interpretations. For instance, in the conjunctive architecture, where the fusion rule is ``disable if at least one supervisor disables'', one can verify the satisfaction of the control objective by tracking the ambiguity between the two local agents in a synchronized product construction. Similarly, in the inference-based architecture, one analyzes an $N$-fold synchronized  product of local information states to resolve ambiguity through $N$ levels of inference.
However, these constructions fundamentally rely on specific monotone structures of the fusion rule $F$, which arise naturally when decision symbols are assigned direct reasoning interpretations. Although this structural property leads to decidability results, such interpretation-driven architectures do not fully exploit the expressive power of decentralized decision-making systems.
For example, even in the binary decision setting, the exclusive-or (XOR) fusion rule is typically not considered, since it lacks a direct physical interpretation such as enabling or disabling an event. At first glance, such a fusion rule may appear unnatural from a control-theoretic perspective, since the event must be globally disabled even when both agents agree to enable it.
However, from a purely mathematical perspective, decision alphabets are simply abstract symbols, and the XOR fusion rule indeed induces a fundamentally different computational structure (as shown in this work).

These observations suggest that treating decision alphabets and fusion rules purely as symbolic objects, beyond their physical interpretation, can lead to alternative results in decentralized decision-making. In fact, the following general but fundamental problem remains largely unexplored:
\begin{resp}
\emph{``Given an arbitrary decentralized architecture $(\Delta, F)$, where $\Delta$ is a finite decision alphabet and $F:\Delta \times \Delta \to \{0,1\}$, does there exist an algorithm to determine the existence of local decision maps $S_i:\Sigma_{o,i}^* \to \Delta$ such that the global specification is achieved?"}
\end{resp}

\noindent
This problem is equivalent to the setting where the fusion rule is also unspecified and must be designed jointly with the local decision maps. 
Since $\Delta$ is finite, the set of all possible fusion rules  is also finite. Therefore, one may in principle enumerate all such fusion rules and check, for each candidate rule, whether there exist local decision maps achieving the specification.

\subsection{Main Result}
In this paper, we show that the decentralized decision-making problem under an arbitrary decentralized architecture with a finite decision alphabet is undecidable. In fact, we establish that undecidability already arises in a highly restricted setting, namely when the local decision alphabet is binary and the fusion rule is fixed to the XOR Boolean operator.
This result is somewhat surprising in light of the existing literature, where most works over the past three decades focus on enriching the decision alphabet or designing more sophisticated fusion rules to enhance the expressive power of decentralized architectures. In contrast, we show that even the simple XOR fusion rule, arguably one of the simplest binary operations beyond AND and OR, already leads to undecidability of the overall decision problem.
However, from a structural viewpoint, this is not unexpected. While AND and OR are monotone operators that preserve a natural ordering structure, XOR is inherently non-monotone and does not admit such an order-preserving interpretation. This fundamental difference eliminates the standard monotonicity-based techniques used in decentralized supervisory control.

Technically, our approach formulates a decentralized observation problem in which the objective is to determine whether the system execution belongs to a given regular sublanguage. We show that this problem is undecidable via a reduction from the Thue word problem, which concerns equivalence of two words under a finite set of rewriting rules.
Furthermore, we establish that this decentralized observation problem can be reduced to classical decentralized control, diagnosis, and prognosis problems. This implies the undecidability of these problems under finite communication alphabets.

Notably, it has previously been shown that the notion of joint diagnosability/observability is undecidable \cite{sengupta2002decentralized,tripakis2004undecidable}. That setting corresponds to allowing an \emph{infinite} communication alphabet, where each supervisor may transmit its entire observation sequence to the fusion center. In that case, the source of undecidability is less structurally explicit, as both the communication alphabet and fusion structure are unbounded.
In contrast, our result is strictly stronger in the sense that undecidability persists even when the fusion rule is fixed and  the communication alphabet is binary. This shows that the intrinsic source of undecidability lies in the decentralized information structure induced by independent observations, rather than in the size or richness of the communication alphabet.

\subsection{Organization}
The rest of the paper is organized as follows. In Section~\ref{sec:problem}, we formulate the decentralized observation problem with a finite decision alphabet, which serves as the central problem for which we establish undecidability, and which further implies undecidability for other decentralized decision-making problems.
In Section~\ref{sec:thue-word-problem}, we introduce the well-known Thue word problem, which serves as the source undecidable problem for our reduction.
In Section~\ref{sec:construction}, we present the reduction from the Thue word problem to the decentralized observation problem. The formal correctness proof of the reduction is provided in Section~\ref{sec:proof}.
In Section~\ref{sec:extensions}, we show that other related decentralized supervisory control, fault diagnosis, and fault prognosis problems are also undecidable under binary decision alphabets and the XOR fusion rule.
Finally, Section~\ref{sec:conclusion} concludes the paper.

\section{Decentralized Observation Problem}
\label{sec:problem}
Let $\Sigma$ be a finite event set (alphabet). A string (or word) is a finite sequence of events, and $\Sigma^*$ denotes the set of all finite strings over $\Sigma$, including the empty string $\varepsilon$. The length of a string $s \in \Sigma^*$ is denoted by $|s|$. 
For two strings $s,t \in \Sigma^*$, their concatenation is denoted by $st$, obtained by appending $t$ after $s$. For two languages $L_1, L_2 \subseteq \Sigma^*$, their concatenation is defined as
$L_1 L_2 = \{ st \mid s \in L_1, t \in L_2 \}.$
For a language $L \subseteq \Sigma^*$ and a string $s \in \Sigma^*$, the post-language of $L$ after $s$ is defined as $L/s = \{ t \in \Sigma^* \mid st \in L \}$, 
which represents all suffixes $t$ such that the concatenation $st$ remains in $L$.
The prefix-closure of a language $L \subseteq \Sigma^*$ is defined by
\[
\mathrm{Pref}(L) = \{ s \in \Sigma^* \mid \exists t \in \Sigma^* \text{ such that } st \in L \}.
\]
The exclusive-or (XOR) operator is denoted by $\oplus$, where for any $a,b\in\{0,1\}$, we have 
 $a \oplus b = 1 \Leftrightarrow a \neq b.$

We consider a discrete-event system modeled by a deterministic finite-state automaton
\begin{equation}
    G = (Q, \Sigma, f, q_0),
\end{equation}
where $Q$ is a finite state set, $q_0 \in Q$ is the initial state, and $f: Q \times \Sigma \to Q$
is a (possibly partial) deterministic transition function. For $q,q'\in Q$ and $\sigma\in\Sigma$, $f(q,\sigma)=q'$ means that there exists a transition from $q$ to $q'$ labeled by $\sigma$.
The transition function is extended to strings in the standard way:
(i) $f(q,\varepsilon)=q$; and (ii)
$f(q,s\sigma)=f(f(q,s),\sigma), \forall s\in\Sigma^*,\ \sigma\in\Sigma$.
We also write $q \xrightarrow{s} q'$ if $f(q,s)=q'$. The language generated by $G$ is defined as $\mathcal{L}(G) = \{ s\in\Sigma^* \mid f(q_0,s)! \}$.

Let $\Sigma' \subseteq \Sigma$. The natural projection $P_{\Sigma'}:\Sigma^* \to {\Sigma'}^*$ is defined recursively by: for every $s\in\Sigma^*$ and $\sigma\in\Sigma$, we have
\begin{equation}
P_{\Sigma'}(\varepsilon)=\varepsilon
\text{ and }
    P_{\Sigma'}(s\sigma)=
    \begin{cases}
        P_{\Sigma'}(s)\sigma, & \text{if } \sigma\in\Sigma',\\
        P_{\Sigma'}(s), & \text{if } \sigma\notin\Sigma'.
    \end{cases}
\end{equation}

In the decentralized setting, the plant is observed by multiple \emph{local agents}, each of which can observe only   a subset of events. For the purpose of establishing undecidability, it suffices to consider two agents. For each $i\in\{1,2\}$, let $\Sigma_{o,i}\subseteq\Sigma$ be the set of locally observable events of agent $i$, and let
\[
P_i:\Sigma^* \to \Sigma_{o,i}^*
\]
be the corresponding natural projection.

Each local agent produces a local decision symbol from a finite communication alphabet $\Delta$ based only on its own observation. 
Thus, a local decision map for agent $i$ is a function
\begin{equation}
    S_i:P_i(\L(G))\to\Delta.
\end{equation}
Equivalently, one may define $S_i$ on the whole set $\Sigma_{o,i}^*$, since its values outside $P_i(\mathcal L(G))$ are irrelevant. 
The two local decisions are then sent to a central fusion site. 
The fusion site is memoryless and is described by a fixed function
\begin{equation}
    F:\Delta\times\Delta\to\{0,1\}.
\end{equation}

The objective of the fusion center is to determine whether a generated string belongs to a given target regular language $K \subseteq \mathcal{L}(G)$. The decentralized observation problem with finite communication alphabet is defined as follows.

\begin{myprob}[\bf Decentralized Observation Problem with Finite Alphabet]
\label{prob:decentralized-observation}
Given a deterministic finite-state automaton $G=(Q,\Sigma,f,q_0)$, locally observable event sets $\Sigma_{o,1},\Sigma_{o,2}\subseteq\Sigma$, a regular language $K \subseteq \mathcal{L}(G)$, a finite communication alphabet $\Delta$, and a fusion rule
\begin{equation}
F:\Delta \times \Delta \to \{0,1\},    
\end{equation}
decide whether there exist local decision maps
\begin{equation}
    S_i: P_i(\mathcal{L}(G)) \to \Delta,\quad i=1,2,
\end{equation}
such that for all $s \in \mathcal{L}(G)$, we have
\begin{equation}
s \in K
\quad \Longleftrightarrow \quad
F\big(S_1(P_1(s)), S_2(P_2(s))\big)=1.
\label{eq:exact-decentralized-decision}
\end{equation}
\end{myprob}

In other words, the problem asks whether the characteristic function of $K$ over $\mathcal{L}(G)$ can be exactly reconstructed using only finite-symbol local communication based on partial observations. Importantly, the local decision maps are arbitrary set-theoretic functions; no finite-state or computability restrictions are imposed beyond dependence on local observations. 
As will be further discussed in Section~\ref{sec:extensions}, the decentralized observation problem plays a fundamental role in a wide range of decentralized decision-making problems, including decentralized supervisory control and decentralized fault diagnosis.
 
For simplicity and without loss of generality, we assume that the specification language $K$ is given in state-based form. Since $K$ is regular, one can always construct a product  of $G$ with an automaton recognizing $K$. Thus, we may assume  the existence of a set of marked states  $Q_K \subseteq Q$ accepting $K\subseteq \mathcal{L}(G)$, i.e., 
 $\forall s \in \mathcal{L}(G)$, $s \in K$ iff $f(q_0,s) \in Q_K.$ 

The main objective of this paper is to show that Problem~\ref{prob:decentralized-observation} is undecidable in general. Moreover, this undecidability holds even under a highly restricted architecture: \textbf{binary communication alphabet and a fixed XOR fusion rule}.

\begin{mythm}
\label{thm:main-undecidability}
The decentralized observation problem with finite communication alphabet is undecidable. In particular, there is no algorithm that, given $G$, $\Sigma_{o,1}$, $\Sigma_{o,2}$, and $K$, decides whether there exist local decision maps
\[
S_1: P_1(\mathcal{L}(G)) \to \{0,1\}\quad \text{ and }\quad 
S_2: P_2(\mathcal{L}(G)) \to \{0,1\},
\]
such that for all $s \in \mathcal{L}(G)$,
\begin{equation}
s \in K
\quad \Longleftrightarrow \quad
S_1(P_1(s)) \oplus S_2(P_2(s)) = 1,
\label{eq:xor-decision}
\end{equation}
where $\oplus$ denotes exclusive-or.
Consequently, the problem remains undecidable even with two agents, binary communication, and a fixed memoryless XOR fusion rule.
\end{mythm}

\begin{remark}[Comparison with existing undecidability results]
In \cite{tripakis2004undecidable}, it was shown that the following decentralized observation problem is undecidable: determine the existence of a memoryless fusion rule
\[
F:\Sigma_{o,1}^* \times \Sigma_{o,2}^* \to \{0,1\},
\]
such that, for any string $s$, we have
\[
s \in K \quad \Longleftrightarrow  \quad F(P_1(s),P_2(s)) = 1.
\]
In that setting, each local agent is allowed to transmit its entire observation history, resulting in messages of unbounded length sent to the fusion center. Therefore, that result corresponds to the case of an unrestricted (infinite) communication alphabet and an unspecified fusion rule.
Our result is strictly stronger, as we fix the fusion rule a priori to be the exclusive-or operator and restrict each local agent to transmit only a single symbol from a binary communication alphabet. Despite these severe restrictions, the problem remains undecidable.
\end{remark}

\begin{remark}[Comparison with existing decidability results]
For the case of binary communication alphabets, it has been shown in the literature that the problem is decidable (polynomial in the number of states but exponential in the number of agents) when the fusion rule is either conjunction or disjunction. This setting corresponds to the classical decentralized supervisory control architecture with unconditional decisions.
Our result does not contradict these decidability results. In fact, conjunction and disjunction are monotone fusion rules and admit natural control-theoretic interpretations such as ``\emph{disable if at least one agent disables}'' or ``\emph{disable only if all agents disable}''. These interpretations enable constructive synthesis procedures for the local decision maps.
In contrast, XOR is non-monotone. This fundamental structural difference is a key reason for the emergence of undecidability in the XOR setting.
\end{remark}

\section{Thue Word Problem}
\label{sec:thue-word-problem}

In this work, we establish the undecidability of
Problem~\ref{prob:decentralized-observation} via a reduction from the Thue word problem, which is a classical undecidable problem in formal language theory~\cite{post1947recursive}. For self-containedness, we briefly review the problem below. For clarity, we use the term \emph{word} to refer to symbol sequences in the Thue system, and \emph{string} to refer to symbol sequences in the DES setting.

Let $\Gamma$ be a finite alphabet. A \emph{Thue system} over $\Gamma$ is a finite set of unordered rewriting rules
\[
\mathcal{R} = \{(\ell_j, r_j)\mid j=1,\dots,m\},
\]
where $\ell_j, r_j \in \Gamma^*$ are finite words. Each rule $(\ell_j,r_j)$ is bidirectional, meaning that occurrences of $\ell_j$ may be replaced by $r_j$ and vice versa. In particular, for any words $x,y \in \Gamma^*$, one may replace an occurrence of $\ell_j$ by $r_j$ in the context $x(\cdot)y$, or replace $r_j$ by $\ell_j$ in the same context.

Formally, we define the one-step Thue rewriting relation
\[
\leftrightarrow_{\mathcal{R}} \subseteq \Gamma^* \times \Gamma^*
\]
as follows. For two words $w,w' \in \Gamma^*$, we write
\[
w \leftrightarrow_{\mathcal{R}} w'
\]
if there exist $x,y \in \Gamma^*$ and a rule $(\ell_j,r_j) \in \mathcal{R}$ such that either
\[
w = x\ell_j y,\quad w' = xr_j y,
\]
or
\[
w = x r_j y,\quad w' = x \ell_j y.
\]
For example, consider $\Gamma\!=\!\{a,b,c\}$ and $\mathcal{R}\!=\!\{(ab,ba),(c,bb)\}$. Then,  $abc \leftrightarrow_{\mathcal{R}} bac$
since the substring $ab$ in $abc$ can be replaced by $ba$.
Also, we have $bac \leftrightarrow_{\mathcal{R}} babb$
since the substring $c$ in $bac$ can be replaced by $bb$.

Let $\equiv_{\mathcal{R}}$ denote the reflexive and transitive closure of $\leftrightarrow_{\mathcal{R}}$. That is, for $w,w' \in \Gamma^*$,
\[
w \equiv_{\mathcal{R}} w'
\]
if and only if $w'$ can be obtained from $w$ by applying a finite sequence of Thue rewriting steps (possibly zero steps). Equivalently, there exists a sequence of words
\[
w = w_0, w_1, \dots, w_k = w'
\]
such that $w_{t-1} \leftrightarrow_{\mathcal{R}} w_t$ for all $t=1,\dots,k$. In this case, $w$ and $w'$ are said to be \emph{Thue equivalent} with respect to $\mathcal{R}$.

The Thue word problem is defined as follows.

\begin{myprob}[\bf Thue Word Problem]
\label{prob:thue}
Given a finite alphabet $\Gamma$, a finite Thue system $\mathcal{R}$ over $\Gamma$, and two words $\alpha,\beta \in \Gamma^*$, decide whether
\[
\alpha \equiv_{\mathcal{R}} \beta.
\]
\end{myprob}

It is well known that the Thue word problem is undecidable; that is, there exists no algorithm that, for every finite Thue system $\mathcal{R}$ and every pair of words $\alpha,\beta \in \Gamma^*$, decides whether $\alpha$ and $\beta$ are Thue equivalent.

In the sequel, an instance of the Thue word problem is denoted by $\mathcal{I} = (\Gamma,\mathcal{R},\alpha,\beta)$.
Our reduction constructs, from any such instance $\mathcal{I}$, an instance of the decentralized observation problem such that   feasible local decision maps exist if and only if $\alpha \not\equiv_{\mathcal{R}} \beta$. 
Therefore, decidability of the decentralized observation problem would imply decidability of the Thue word problem, yielding a contradiction.
  
\section{Construction of the Decentralized Observation Problem}
\label{sec:construction}

Given an instance of the Thue word problem
\[
\mathcal I=(\Gamma,\mathcal R,\alpha,\beta),
\]
we construct an instance of Problem~\ref{prob:decentralized-observation}
\[
\mathcal P(\mathcal I)=(G,\Sigma_{o,1},\Sigma_{o,2},K,\Delta,F),
\]
where the communication alphabet is binary, i.e., $\Delta=\{0,1\}$, and the fusion rule is fixed as $F(a,b)=a\oplus b$. Throughout the construction, we distinguish between the Thue alphabet $\Gamma$ and the plant event set $\Sigma$.

The overall structure of the construction is conceptually illustrated in Fig.~\ref{fig:proof}. 
We next present the construction details and the underlying intuition.

\begin{figure*}[tp]
	\centering
	\includegraphics[width=0.94\linewidth]{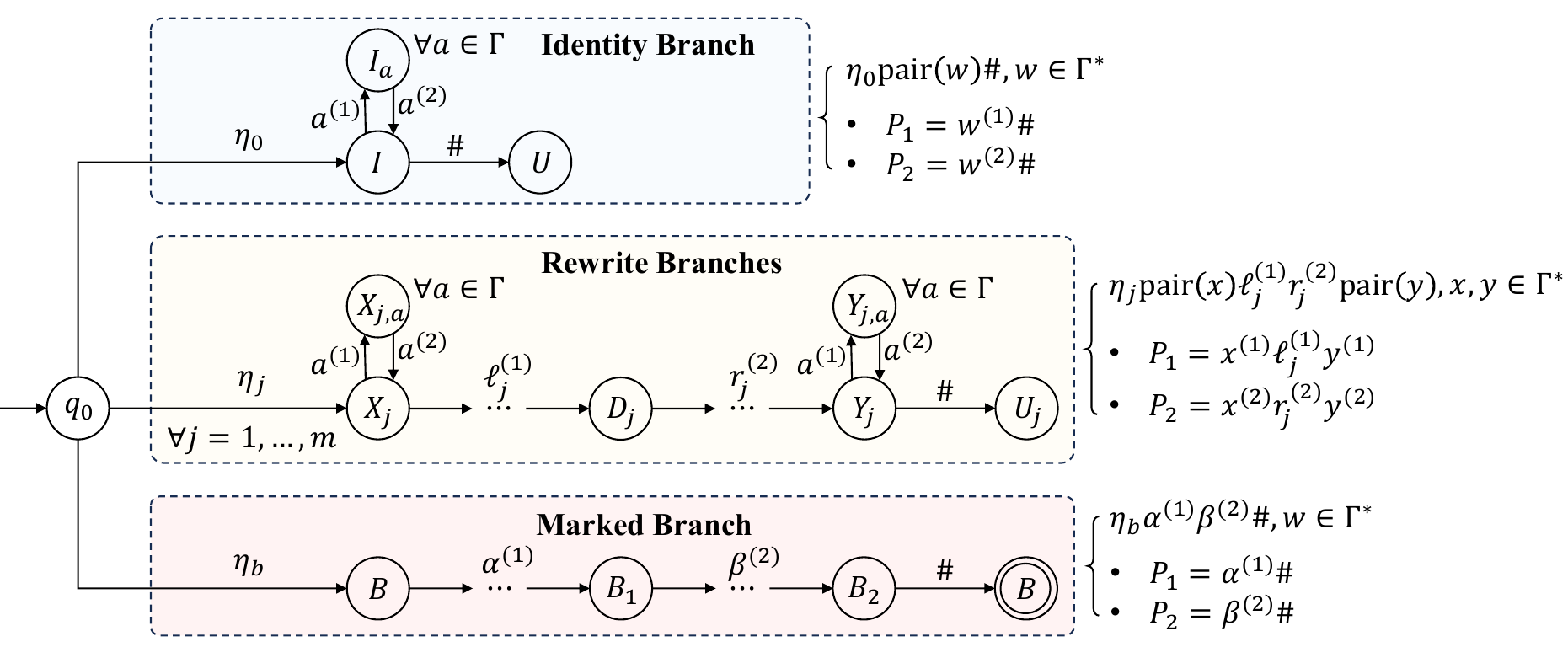}
	\caption{Conceptual illustration of the construction of  a  decentralized observation problem instance $(G,\Sigma_{o,1},\Sigma_{o,2},K,\Delta,F)$ from a  Thue word problem instance $\mathcal I=(\Gamma,\mathcal R,\alpha,\beta)$, where the unique marked state is $Q_K=\{B\}$. }
	\label{fig:proof}
\end{figure*}

\subsection{Event Set and Projections}

The event set of the constructed plant is defined as
\[
\Sigma
=
\Sigma^{(1)} \cup \Sigma^{(2)} \cup \{\#\} \cup \Sigma_{\mathrm{hid}},
\]
where
\begin{itemize}
    \item $\Sigma^{(1)}=\{a^{(1)} \mid a\in\Gamma\}$ (observable only by agent 1),
    \item $\Sigma^{(2)}=\{a^{(2)} \mid a\in\Gamma\}$ (observable only by agent 2),
    \item $\#$ is a termination symbol observable by both agents,
    \item $\Sigma_{\mathrm{hid}}=\{\eta_0,\eta_b,\eta_1,\dots,\eta_m\}$ is a set of events unobservable to both agents.
\end{itemize}
Hence,  for each agent $i\in\{1,2\}$, the observable event set is
\[
\Sigma_{o,i}=\{a^{(i)} \mid a\in\Gamma\}\cup\{\#\},
\]
and the natural projection  $P_i:\Sigma^*\to\Sigma_{o,i}^*$ for agent $i$ is defined accordingly.

For a word $w=a_1\cdots a_k\in\Gamma^*$, we define 
two copies
\[
w^{(1)}=a_1^{(1)}\cdots a_k^{(1)}, \qquad
w^{(2)}=a_1^{(2)}\cdots a_k^{(2)},
\]
and define the paired encoding  
\[
\mathrm{pair}(w)=a_1^{(1)}a_1^{(2)}\cdots a_k^{(1)}a_k^{(2)}.
\]
Then for each word $w\in \Gamma^*$ and each agent $i=1,2$, we have $P_i(\mathrm{pair}(w))=w^{(i)}$, 
i.e.,  only  its own copy is observed.

\subsection{Plant Construction}

Now, we define the deterministic finite-state automaton that is observed by two agents
\[
G=(Q,\Sigma,f,q_0),
\]
where the state set is partitioned as
\[
Q=\{q_0\}\cup Q_{\mathrm{id}}\cup Q_{\mathrm{rew}}\cup Q_{\mathrm{mark}}.
\]
The transition function as well as the state space are defined via three independent branches as follows.

\subsubsection*{\bf 1) Identity Branch}

The identity branch is initially activated by the hidden transition
\[
q_0 \xrightarrow{\eta_0} I.
\]
At state $I$, for each $a\in\Gamma$, we introduce a pair-loop
\[
I \xrightarrow{a^{(1)}} I_a \xrightarrow{a^{(2)}} I,
\]
which generates arbitrary paired words $\mathrm{pair}(w)$ for $w\in\Gamma^*$.
Hence, there are in total $|\Gamma|$ additional states $I_a$ implementing all loops $a^{(1)}a^{(2)},a\in\Gamma$ at state $I$. 
This branch is then terminated by a transition observable to both agents
\[
I \xrightarrow{\#} U,
\]
where $U$ is an \emph{unmarked} terminal state.

Hence this branch generates the prefix-closed language
\[
L_0=\operatorname{Pref}\big(\{\eta_0\,\mathrm{pair}(w)\,\# \mid w\in\Gamma^*\}\big).
\]
Specifically,  for any complete string $\eta_0\,\mathrm{pair}(w)\,\#$ in this branch, the local observations are
\[
P_1(\cdot) = w^{(1)}\# \quad \text{ and }\quad  P_2(\cdot) = w^{(2)}\#.
\]

\subsubsection*{\bf 2) Rewrite Branches}

For each rewrite rule $(\ell_j,r_j)\in\mathcal R,j=1,\dots,m$, a rewrite branch is activated by transition
\[
q_0 \xrightarrow{\eta_j} X_j.
\]
Then the following transitions are defined.
\begin{itemize}
\item \emph{Left Context Generation:} 
At state $X_j$, we generate arbitrary  left context  $x\in\Gamma^*$ via pair-loops using the same construction at state $I$, i.e., 
    \[
    X_j \xrightarrow{a^{(1)}} X_{j,a} \xrightarrow{a^{(2)}} X_j,\forall a\in \Gamma.
    \]
\item \emph{Rewrite Core:} 
After the left-context generation, the automaton executes a deterministic path that realizes the fixed word $\ell_j^{(1)} r_j^{(2)}$. For readability, we use macro-transition notation:
\[
X_j \xrightarrow{ \ell_j^{(1)}} D_j\xrightarrow{r_j^{(2)}}   Y_j
\]
to denote a finite sequence of standard transitions consuming the symbols of $\ell_j^{(1)}$ and $r_j^{(2)}$ sequentially.
This macro-transition corresponds to a finite chain of intermediate states, which are omitted for clarity.
\item 
\emph{Right Context Generation:} At the resulting state $Y_j$, we again generate arbitrary  right context $y\in\Gamma^*$ via pair-loops, i.e., 
    \[
    Y_j \xrightarrow{a^{(1)}} Y_{j,a} \xrightarrow{a^{(2)}} Y_j,\forall a\in \Gamma.
    \]
\end{itemize}
Finally, we terminate this branch by  transition $Y_j \xrightarrow{\#} U_j$, 
where $U_j$ is an \emph{unmarked} terminal state.
We denote by $Q_{rew}$ the set of all states involved in this branch.

Hence this branch generates the prefix-closed language
\[
L_j=\operatorname{Pref}(\{ \eta_j\,\mathrm{pair}(x)\,\ell_j^{(1)}r_j^{(2)}\,\mathrm{pair}(y)\,\#\mid x,y\in\Gamma^*\}).
\]
For any complete string $s_j=\eta_j\,\mathrm{pair}(x)\,\ell_j^{(1)}r_j^{(2)}\,\mathrm{pair}(y)\,\#$ in this branch, the local observations are
\[
P_1(s_j) = (x\ell_j y)^{(1)}\# \quad \text{ and }\quad  P_2(s_j) = (x r_j y)^{(2)}\#.
\]

\subsubsection*{\bf 3) Marked Branch}

The marked branch is activated by
\[
q_0 \xrightarrow{\eta_b} B_0,
\]
followed by a deterministic path
\[
B_0 \xrightarrow{\alpha^{(1)}} B_1 \xrightarrow{\beta^{(2)}} B_2 \xrightarrow{\#} B.
\]
Still, we do not explicitly enumerate all intermediate states along this path, and instead denote by $Q_{\mathrm{mark}}$ the set of all states involved in this branch.

Hence this branch generates the prefix-closed language
\[
L_m=\operatorname{Pref}(\{  \eta_b\,\alpha^{(1)}\beta^{(2)}\# \}).
\]
We define state $B$ as the \textbf{unique marked state}, i.e., $Q_K=\{B\}$, and thus the only marked string is
\[
K=\{s_m=\eta_b\,\alpha^{(1)}\beta^{(2)}\#\}.
\]
For this string,  the local observations are 
\[
P_1(s_m) = \alpha^{(1)}\#\quad \text{ and }\quad  P_2(s_m) = \beta^{(2)}\#.
\]
 
\subsection{Intuition of the Construction}

We now provide an intuitive explanation of the construction and how the three branches enforce the desired logical structure.
Based on the constructed  $G$ and the local observable events, we aim to find two local decision maps $S_1$ and $S_2$ such that
\[
S_1(P_1(s)) \oplus S_2(P_2(s)) = 1
\quad \Longleftrightarrow \quad
s = \eta_b\,\alpha^{(1)}\beta^{(2)}\#,
\]
i.e., the string $s$ reaches the unique marked state $B$.

First, observe that every marked behavior must terminate with the shared event $\#$. Therefore, the essential decisions are those made after the agents observe the termination symbol $\#$. For those observations that do not end with $\#$, the local decision values can be chosen consistently, for instance both equal to $0$, so that the fusion center outputs $0$ on all non-terminal prefixes, which are unmarked. Hence, the correctness of the local decision maps is only determined by their values on terminal observations ending with $\#$. We now analyze the terminal strings generated by each branch.

The purpose of the identity branch is to ensure consistency between the two local encodings of the same word. Specifically, for every word $w\in\Gamma^*$, this branch generates a string of the form $\eta_0\,\mathrm{pair}(w)\,\#$,
where agent 1 observes $w^{(1)}\#$ and agent 2 observes $w^{(2)}\#$.
Since the terminal state $U$ in this branch is unmarked, the correctness condition requires the fusion output to be $0$.
Because XOR equals zero if and only if its two inputs are equal, both agents must assign the same decision value to the corresponding observations. Therefore, we can define a Boolean valuation over words $C:\Gamma^*\!\to\!\{0,1\}$. 
Intuitively, $C(w)$ records the common decision value assigned to the two local copies of the same word $w$.

The purpose of the rewrite branches is to enforce that the valuation $C(\cdot)$ is invariant under Thue rewriting rules. Consider any rule $(\ell_j,r_j)\in\mathcal R$ and arbitrary contexts $x,y\in\Gamma^*$. Due to the pair-loops at states $X_j$ and $Y_j$, the rewrite branch can generate arbitrary left and right contexts. For the complete string generated by this branch, agent 1 observes $P_1(s)=(x\ell_j y)^{(1)}\#,$ 
while agent 2 observes $P_2(s)=(x r_j y)^{(2)}\#.$
Thus, the two local observations correspond to two words that are equivalent under one Thue rewriting step $x\ell_j y \leftrightarrow_{\mathcal R} x r_j y$.
Since all terminal states of the rewrite branches are unmarked, the correctness condition again requires the fusion output to be $0$.  
Therefore, the two local decisions must be equal. Using the definition of $C$, this gives $C(x\ell_j y)=C(x r_j y).$ 
Since $x$ and $y$ are arbitrary and every rule in $\mathcal R$ has a corresponding rewrite branch, the valuation $C(\cdot)$ must be constant along every one-step Thue rewrite. Consequently, $C$ is constant on each Thue equivalence class induced by $\equiv_{\mathcal R}$.

Finally, the marked branch enforces separation between the two distinguished words $\alpha$ and $\beta$ in the Thue word problem instance.  
Specifically, this branch generates the unique marked string $s_m=\eta_b\,\alpha^{(1)}\beta^{(2)}\#$, 
where agent 1 observes $P_1(s_m)=\alpha^{(1)}\#,$ 
and agent 2 observes
$P_2(s_m)=\beta^{(2)}\#$. 
Since the terminal state of this branch is marked,  
the correctness condition requires the fusion output to be $1$.  
Therefore, the two local decisions must be different. Using the definition of $C$, this gives $C(\alpha)\neq C(\beta).$ 
However, the rewrite branches already force $C$ to be constant on each Thue equivalence class, this  contradicts the requirement imposed by the marked branch if $\alpha \equiv_{\mathcal{R}} \beta$.

Thus, for the constructed instance of the decentralized observation problem, a feasible pair of local decision maps exists if and only if $\alpha$ and $\beta$ are not Thue equivalent. 
This establishes a direct correspondence between the feasibility of decentralized observation and the Thue word problem.

\section{Detailed Proof of Undecidability}
\label{sec:proof}

In this section, we complete the reduction by formally proving the correctness of the construction. We show that the constructed decentralized observation instance admits a feasible solution if and only if the corresponding Thue instance does not satisfy $\alpha \equiv_{\mathcal{R}} \beta$.

\begin{mythm}[Correctness of Reduction]
For the constructed instance $\mathcal{P}(\mathcal{I})$, there exist local decision maps $S_1,S_2$ satisfying~\eqref{eq:xor-decision} if and only if
\[
\alpha \not\equiv_{\mathcal{R}} \beta.
\]  
\end{mythm}

We prove both directions separately.
 
\subsection*{(1) Necessity: If a feasible solution exists, then $\alpha \not\equiv_{\mathcal{R}} \beta$}

Assume that for the constructed decentralized observation problem, 
there exist local decision maps $S_1,S_2$ satisfying
\begin{equation}\label{eq:dec-req}
    S_1(P_1(s)) \oplus S_2(P_2(s)) = 1
\quad \Longleftrightarrow \quad
s \in K.
\end{equation}
First, we define an induced global valuation $C:\Gamma^* \to \{0,1\}$ by the identity branch:
\begin{equation}\label{eq:C-func}
    C(w) := S_1(w^{(1)}\#) = S_2(w^{(2)}\#), \quad \forall w \in \Gamma^*.
\end{equation}
This function is well-defined because the identity branch generates only unmarked strings. 
Specifically,  for every $w$, since 
\begin{equation}
    \eta_0\operatorname{pair}(w)\#\notin K, 
\end{equation}
according to Eq.~\eqref{eq:dec-req},  we have 
\begin{equation}
    S_1(w^{(1)}\#)\oplus S_2(w^{(2)}\#)=0, 
\end{equation}
which implies that 
\begin{equation}
    S_1(w^{(1)}\#)=S_2(w^{(2)}\#).
\end{equation}

Next, consider any rewrite branch corresponding to a rule $(\ell_j,r_j)\in\mathcal{R}$. 
Let $s$ be an arbitrary complete string generated by this branch.
By construction of the rewrite branch, any such string $s_j$ has the form
\begin{equation}
    s_j = \eta_j\,\mathrm{pair}(x) \,\ell_j^{(1)}r_j^{(2)}\,\mathrm{pair}(y)\,\#
\quad \text{for some } x,y\in\Gamma^*.
\end{equation}
Its projections to local agents  are
\begin{equation}
P_1(s_j) = (x\ell_j y)^{(1)}\#\quad\text{ and }\quad  P_2(s_j) = (x r_j y)^{(2)}\#.
\end{equation}
Since all such strings lead to unmarked states, i.e., $s_j\notin K$, the correctness condition of the decentralized observation problem in Eq.~\eqref{eq:dec-req} requires that
\begin{equation}
    S_1((x\ell_j y)^{(1)}\#) \oplus  S_2((x r_j y)^{(2)}\#) = 0. 
\end{equation}
Substituting into the XOR condition yields
\begin{equation}
    S_1((x\ell_j y)^{(1)}\#) = S_2((x r_j y)^{(2)}\#).
\end{equation}
According to Eq.~\eqref{eq:C-func},  
we know that 
\begin{equation}\label{eq:implies-C}
C(x\ell_j y) =S_1((x\ell_j y)^{(1)}\#)= S_2((x r_j y)^{(2)}\#)= C(x r_j y).    
\end{equation}
Since  context words $x,y\in \Gamma^*$ and rule $(\ell_j,r_j)\in\mathcal{R}$ are arbitrary, 
$C$ is invariant under all Thue rewriting rules, and therefore the value of $C$ is constant on each equivalence class induced by $\equiv_{\mathcal{R}}$, i.e., 
\begin{equation}
    \forall w,w'\in \Gamma^*: w\equiv_{\mathcal{R}}w' \Rightarrow C(w)=C(w').
\end{equation}

Finally, consider the marked branch. Since the terminal state $B$  is the unique marked state, 
we have 
\begin{equation}
    \eta_b \alpha^{(1)}\beta^{(2)}\#\in K, 
\end{equation}
where the projections to agents 1 and 2 are $\alpha^{(1)}\#$ and $\beta^{(2)}\#$, respectively. 
Still, according to Eq.~\eqref{eq:dec-req}, we must have 
\begin{equation}
    S_1(\alpha^{(1)}\#)\oplus S_2(\beta^{(2)}\#)=1. 
\end{equation}
Following the same argument in Eq.~\eqref{eq:implies-C}, we  have  
\begin{equation}
    C(\alpha)\neq C(\beta).
\end{equation}
However, if $\alpha \equiv_{\mathcal{R}} \beta$, then invariance of $C$ would imply
$C(\alpha)=C(\beta)$,  which contradicts the above requirement. Therefore, we have $\alpha \not\equiv_{\mathcal{R}} \beta$, which completes the proof.

\subsection*{(2) Sufficiency: If $\alpha \not\equiv_{\mathcal{R}} \beta$, then a feasible solution exists}

Assume now that $\alpha \not\equiv_{\mathcal{R}} \beta$.
We denote by 
\begin{equation}
    [\alpha]_{\mathcal{R}}=\{w\in\Gamma^*\mid w \equiv_{\mathcal{R}} \alpha\}
\end{equation}
the equivalence class of word $\alpha$ under relation $\mathcal{R}$.
We define a Boolean valuation $C:\Gamma^* \to \{0,1\}$ on Thue equivalence classes as follows:
\begin{equation}
    C(w) =
\begin{cases}
0, & w \in [\alpha]_{\mathcal{R}},\\
1, & w\notin [\alpha]_{\mathcal{R}}.
\end{cases}
\end{equation}
We now construct two local decision maps by:
\begin{equation}\label{eq:local-def}
    S_1(w^{(1)}\#) := C(w)
\quad\text{ and }\quad
S_2(w^{(2)}\#) := C(w).
\end{equation}
For all other observation strings not ending with $\#$, 
we define arbitrary values (e.g., $0$) for both agents, since they only correspond to prefixes and do not affect terminal decisions.

Now, we verify that these two local decision maps defined in Eq.~\eqref{eq:local-def} solve the decentralized observation problem. 
We only need to evaluate terminal strings in each branch 
as the two local agents always have the same decision for those non-terminal prefixes by construction, which means that the global decision is $0$ under the XOR rule. 
\begin{itemize}
\item \textbf{Identity branch:} 
For any word $w\in \Gamma^*$, we have
\begin{equation}
    S_1(w^{(1)}\#)\oplus S_2(w^{(2)}\#)=C(w)\oplus C(w)=0,
\end{equation}
which is consistent with unmarked states.

\item \textbf{Rewrite Branches:}
If string $s\in \mathcal{L}(G)$ is generated by a rewrite branch corresponding to $(\ell_j,r_j)\in\mathcal{R}$, then by construction there exist $x,y\in\Gamma^*$ such that
\begin{equation}
    P_1(s) = (x\ell_j y)^{(1)}\#\quad\text{ and }\quad 
P_2(s) = (x r_j y)^{(2)}\#.
\end{equation}
Since $(\ell_j,r_j)\in\mathcal{R}$, we have the one-step Thue relation
$x\ell_j y \leftrightarrow_{\mathcal{R}} x r_j y$,  
which further implies
\begin{equation}
    x\ell_j y \equiv_{\mathcal{R}} x r_j y.
\end{equation}
Because $C$ is defined on equivalence classes of $\equiv_{\mathcal{R}}$, it follows that
\begin{equation}
    C(x\ell_j y) = C(x r_j y).
\end{equation}
Therefore, we know that
\begin{equation}
    S_1(P_1(s)) \oplus S_2(P_2(s))
= C(x\ell_j y)\oplus C(x r_j y)
= 0,
\end{equation}
which satisfies the correctness condition for all terminal rewrite-branch strings.

\item \textbf{Marked Branch:}
For the unique terminal and marked string in this branch
\begin{equation}
    s_m = \eta_b\,\alpha^{(1)}\beta^{(2)}\#,
\end{equation}
the local observations are
\begin{equation}
    P_1(s_m)=\alpha^{(1)}\#\quad\text{ and }\quad  P_2(s_m)=\beta^{(2)}\#.
\end{equation}
Then by Eq.~\eqref{eq:local-def}, we have 
\begin{equation}
    S_1(P_1(s_m)) \oplus S_2(P_2(s_m))
= C(\alpha)\oplus C(\beta).
\end{equation}
Since $\alpha \not\equiv_{\mathcal{R}} \beta$, they belong to different equivalence classes. Hence,
\begin{equation}
C(\alpha)\neq C(\beta),
\end{equation}
and therefore
\begin{equation}
S_1(P_1(s_m)) \oplus S_2(P_2(s_m)) = 1,
\end{equation}
which satisfies the requirement for the marked branch.
\end{itemize}
Thus all terminal strings ending up with $\#$  satisfy the requirement in the decentralized observation problem, which means that local decision mappings $S_1,S_2$ constructed in Eq.~\eqref{eq:local-def} are a valid solution. 

\begin{remark}
Although the classical Thue word problem is defined via symmetric rewrite rules, in our reduction it is sufficient to consider only one direction of each rule. That is, we only construct path of form \(\ell_j^{(1)}r_j^{(2)}\) in the rewrite branch. The reason is that the constructed constraint enforces equality preservation under rewriting. Such an equality constraint is inherently symmetric, so adding the inverse rule by path $r_j^{(1)}\ell_j^{(2)} $ would impose no additional restriction beyond what is already enforced.  
\end{remark}

\begin{remark}[Why the reduction is specific to XOR] 
The key step in the reduction is the identity branch, which forces
\[
S_1(w^{(1)}\#)\oplus S_2(w^{(2)}\#)=0
 \Rightarrow 
S_1(w^{(1)}\#)=S_2(w^{(2)}\#),
\]
allowing the definition of a single well-defined valuation \(C(w)\). This collapse into a shared scalar labeling is essential, since the rewrite branches then enforce invariance
$C(x\ell_j y)=C(xr_j y)$
so that the Thue congruence is encoded as equality constraints on a single function.

For AND/OR fusion rules, this collapse step fails. For example, under AND rule, that
\[
S_1(w^{(1)}\#)\wedge S_2(w^{(2)}\#)=0
\]
does not imply equality of the two local decisions; it only excludes the pair \((1,1)\). Hence the identity branch does not induce a well-defined common valuation \(C(w)\), and the rewrite branches cannot be interpreted as invariance constraints on a single labeling of words.
For the OR rule, the condition
\[
S_1(w^{(1)}\#)\vee S_2(w^{(2)}\#)=0
\]
does force $S_1(w^{(1)}\#)=S_2(w^{(2)}\#)=0$ for every word \(w\). However, this collapse is degenerate: it forces all such terminal observations to have value \(0\). In particular, the marked branch would then also satisfy $S_1(\alpha^{(1)}\#)\vee S_2(\beta^{(2)}\#)=0\vee0=0$, so the required separation on marked state cannot be achieved.
Therefore, the proof is not merely using decentralized  observations; it specifically exploits the fact that XOR distinguishes equality from inequality between the two local decisions. This equality/disequality structure is absent in the classical AND and OR fusion rules.
\end{remark}
\section{Undecidability in Decentralized Control, Diagnosis, and Prognosis}\label{sec:extensions}
In this section, we further extend the undecidability results from the decentralized observation problem to more specific decentralized decision-making problems, including supervisory control, fault diagnosis, and fault prognosis. We show that the decentralized observation problem can be  reduced to each of these settings. Consequently, the existence of local decision maps in these specific decentralized frameworks is also undecidable.

\subsection{Undecidability in Decentralized Supervisory Control}
 
In the context of decentralized supervisory control \cite{cieslak1988supervisory,rudie1992think,yoo2002general}, each local agent (i.e., a local supervisor) is associated with its own set of controllable events \(\Sigma_{c,i} \subseteq \Sigma\). Accordingly, each supervisor generates a local control action depending on its observation history and the event to be controlled, i.e.,
\[
S_i: \Sigma_{o,i}^* \times \Sigma_{c,i} \to \Delta,
\]
where \(\Delta\) is a finite decision alphabet.

Since we consider a two-agent setting, the fusion rule \(F:\Delta \times \Delta \to \{0,1\}\) is applied only when a controllable event is jointly interpreted by both supervisors. For events controlled exclusively by one agent, i.e., \(\sigma \in \Sigma_{c,i} \setminus \Sigma_{c,j}\), the corresponding symbol in \(\Delta\) already encodes the control decision. To extract the actual control action, we introduce a decoding function $D:\Delta \to \{0,1\},$ 
which maps each local decision symbol to its corresponding enable/disable action. This decoding can be regarded as an implicit component of the fusion architecture.
We note that such a decoding function arises naturally in classical conjunctive (AND) and disjunctive (OR) architectures, where decision symbols are directly interpreted as control commands (e.g., ``enable" or ``disable"), and thus no additional decoding layer is required.

The global supervisory control map is therefore defined as
\[
S_{1,2}^F : \mathcal{L}(G) \times \Sigma \to \{0,1\},
\]
where for any \(s \in \mathcal{L}(G)\) and \(\sigma \in \Sigma\), we have
\begin{align}
&S_{1,2}^F(s,\sigma) =   \\
&\begin{cases}
1, & \sigma \notin \Sigma_{c,1} \cup \Sigma_{c,2},\\
D(S_1(P_1(s),\sigma)), & \sigma \in \Sigma_{c,1} \setminus \Sigma_{c,2},\\
D(S_2(P_2(s),\sigma)), & \sigma \in \Sigma_{c,2} \setminus \Sigma_{c,1},\\
F\big(S_1(P_1(s),\sigma),S_2(P_2(s),\sigma)\big), & \sigma \in \Sigma_{c,1} \cap\Sigma_{c,2}.
\end{cases}\nonumber
\end{align}
That is, uncontrollable events are always enabled (decision 1), while controllable events are enabled according to the fused decision.

The corresponding closed-loop language under \(S_{1,2}^F\) is defined inductively as follows:
\begin{itemize}
    \item 
    $\varepsilon \in \mathcal{L}(S_{1,2}^F/G)$; 
    \item 
    for any \(s \in \Sigma^*,\sigma \in \Sigma\), 
    we have $s\sigma \in \mathcal{L}(S_{1,2}^F/G)$ 
    iff 
      (i) $s \in \mathcal{L}(S_{1,2}^F/G)$; 
     (ii) $s\sigma \in \mathcal{L}(G)$; and 
    (iii) $S_{1,2}^F(s,\sigma)=1$.
\end{itemize} 

The \textbf{decentralized supervisory control problem}  is to decide the existence of local supervisory control laws \(S_i\) such that the resulting closed-loop behavior exactly matches a given specification language \(L_{\mathrm{spec}} \subseteq \mathcal{L}(G)\), i.e.,
\[
\mathcal{L}(S_{1,2}^F / G) = L_{\mathrm{spec}}.
\]

Hereafter, we show that the decentralized supervisory control problem is also undecidable, even when restricted to binary local decisions \(\Delta=\{0,1\}\) and the XOR  fusion rule.

\begin{mythm}
Let \((G,\Sigma_{o,1},\Sigma_{o,2},K,\Delta,F)\) be an instance of the decentralized observation problem. Then there exists a polynomial-time reduction from this instance to an instance of the decentralized supervisory control problem. Consequently, the decentralized supervisory control problem is undecidable, even under binary decision alphabets and XOR fusion.
\end{mythm}
\begin{proof}
We construct a decentralized supervisory control instance 
\((\tilde G,\tilde\Sigma_{c,1},\tilde\Sigma_{c,2},L_{\mathrm{spec}})\) 
such that 
\(\mathcal{L}(\tilde G)=\mathcal{L}(G)\cup \{s\sigma_c \mid s\in \mathcal{L}(G)\}\), 
    \(\tilde\Sigma_{c,1}=\tilde\Sigma_{c,2}=\{\sigma_c\}\),
    and \(L_{\mathrm{spec}}=\mathcal{L}(G)\cup \{s\sigma_c \mid s\in K\}\). 
The fusion rule used is still the XOR operator $\oplus$.

$(\Rightarrow)$
Assume there exist two local supervisors \(S_1,S_2\) such that
$\mathcal L(S_{1,2}^F/\tilde G)=L_{\mathrm{spec}}$.
We can define induced maps:
\[
\hat S_i(P_i(s)) := S_i(P_i(s),\sigma_c).
\]
Then for any \(s\in \mathcal{L}(G)\), we have
\begin{align}
 s\sigma_c \in \mathcal L(S_{1,2}^F/\tilde G)  
 &\Leftrightarrow
S_{1,2}^F(s,\sigma_c)=1\nonumber\\
& \Leftrightarrow
\hat S_1(P_1(s)) \oplus \hat S_2(P_2(s))=1\nonumber\\
& \Leftrightarrow
s\sigma_c \in L_{\mathrm{spec}}\nonumber\\
 &\Leftrightarrow
 s\in K.\nonumber 
\end{align}
Hence, \(\hat S_1,\hat S_2\) solve the decentralized observation problem.

$(\Leftarrow)$ 
Assume there exist two local decision maps \(S_1,S_2\) solving the decentralized observation problem, i.e.,
\[
S_1(P_1(s)) \oplus S_2(P_2(s))=1
\Leftrightarrow
s\in K.
\]
Construct supervisors for the control problem by:
\[
S_i^{ctrl}(P_i(s),\sigma_c):=S_i(P_i(s)).
\]
Then the fused control becomes:
\[
S_{1,2}^F(s,\sigma_c)=
S_1(P_1(s)) \oplus S_2(P_2(s)).
\]
Thus, we have 
\[
S_{1,2}^F(s,\sigma_c)=1
\Leftrightarrow
s\in K
\Leftrightarrow
s\sigma_c \in L_{\mathrm{spec}}.
\]
Since all other events are uncontrollable, they are always enabled, hence do not affect the closed-loop language. Therefore, $\mathcal L(S_{1,2}^F/\tilde G)=L_{\mathrm{spec}}.$ 
\end{proof}

\subsection{Undecidability in Decentralized Fault Diagnosis Problem}
In the context of decentralized fault diagnosis \cite{debouk2000coordinated,wang2007diagnosis}, we assume that the plant contains a fault event \(\sigma_F \in \Sigma\). We define the fault language as
\[
\Psi(\sigma_F) := \Sigma^* \{\sigma_F\} \  \cap \mathcal{L}(G),
\]
i.e., the set of all strings generated by the plant that end   with the fault event \(\sigma_F\). 
We denote by $\sigma_F\in s$ that the fault event occurs in string $s\in \Sigma^*$. 

Each local agent acts as a diagnoser and produces a local decision based on its observation history. Thus, each diagnoser is a mapping of the form $S_i : \Sigma_{o,i}^* \to \Delta$, 
where \(\Delta\) is a finite decision alphabet. Given a fusion rule \(F:\Delta \times \Delta \to \{0,1\}\), the objective of the \textbf{decentralized fault diagnosis problem} is to decide the existence of local diagnosers
$S_i$ such that:
\begin{itemize}
    \item[1)]
    Any  occurrence of the fault event can be detected within a bounded delay, i.e., 
    \begin{align}
   &(\exists n\in\mathbb{N})(\forall s\in \Psi(\sigma_F))(\forall t \in \mathcal{L}(G)/s:|t| \geq n)\nonumber \\
   &\text{s.t. }F\big(S_1(P_1(st)),S_2(P_2(st))\big)=1.\nonumber
\end{align}
    \item[2)] 
    No false alarm is made when the fault event does not actually occur, i.e., for any $s\in \mathcal{L}(G)$, we have
    \begin{equation}
         \sigma_F\not\in s\Rightarrow F\big(S_1(P_1(s)),S_2(P_2(s))\big)=0.
    \end{equation}
\end{itemize}

We show that this decentralized fault diagnosis problem is also undecidable  when the local decision alphabet is finite.
\begin{mythm}
Let \((G,\Sigma_{o,1},\Sigma_{o,2},K,\Delta,F)\) be an instance of the decentralized observation problem. Then there exists a polynomial-time reduction from this instance to an instance of the decentralized fault diagnosis problem. Consequently, the decentralized fault diagnosis  problem is undecidable, even under binary decision alphabets and XOR fusion.
\end{mythm}
\begin{proof}
We construct a decentralized diagnosis instance \((\tilde G,\sigma_F)\) from an instance of the decentralized observation problem as follows.

We extend the plant \(G=(Q,\Sigma,f,q_0)\) by introducing three new events 
$\{\sigma_F,\sigma_N,\tau\}$,  where
\begin{itemize}
    \item 
    \(\sigma_F\): fault event (unobservable to both agents), 
    \item 
    \(\sigma_N\): non-fault event (also unobservable),
    \item 
    \(\tau\): padding event (observable to both agents).
\end{itemize}
Let \(K \subseteq \mathcal L(G)\) be the specification language from the observation problem.
We define the new plant language:
\begin{equation}
    \tilde{\mathcal L}(G)=\mathcal L(G)
\cup
\big(\mathcal L(G)\setminus K\big)\{\sigma_N\}\{\tau\}^*
\cup K \{\sigma_F\}\{\tau\}^*. 
\end{equation}
Let \(Q_K \subseteq Q\) be the set of states accepting language \(K\) in $G$.
The new plant \(\tilde G\) is actually constructed from $G$ by:
\begin{itemize}
    \item 
    adding a new absorbing state \(q_{\mathrm{new}}\) with a self-loop transition
    $ q_{\mathrm{new}} \xrightarrow{\tau} q_{\mathrm{new}}.$
    \item 
    for each \(q \in Q_K\), add transition 
   $q \xrightarrow{\sigma_F} q_{\mathrm{new}}$,
   \item 
    for each \(q \in Q \setminus Q_K\), add transition $q \xrightarrow{\sigma_N} q_{\mathrm{new}}$.
\end{itemize}
We extend observable alphabets by $\tilde\Sigma_{o,i} = \Sigma_{o,i} \cup \{\tau\}$. 

Next, we show that the solvability of the original decentralized observation problem  and 
the constructed decentralized fault diagnosis problem are equivalent.

($\Rightarrow$)
Assume that there exist local decision maps \(S_1,S_2\) solving the decentralized observation problem.
We construct local diagnosers \(\tilde S_i\)  as follows 
\begin{equation}
 \tilde S_i(v_i)=
\begin{cases}
S_i(u_i), & \text{if } v_i=u_i\tau^k \quad u_i\in P_i(\mathcal L(G)),k\geq 1,\\
0, & \text{otherwise}.
\end{cases}   
\end{equation}
That is, before the first occurrence of the observable padding event \(\tau\), both diagnosers output \(0\). After \(\tau\) is observed, each diagnoser applies the original observation map to the observation prefix generated by the original plant \(G\).
Hence the fault is detected after at most one observable \(\tau\)-transition following the occurrence of \(\sigma_F\). Thus the bounded-delay detection condition holds with delay bound \(n=1\).
For observation of string $s\sigma_N\tau$, since \(s\notin K\), the original decentralized observation property implies that the right-hand side equals \(0\). Therefore, no non-faulty execution produces a fused alarm.
Consequently, \(\tilde S_i\) satisfy both bounded-delay fault detection and absence of false alarms. Hence a valid decentralized fault diagnoser exists.

($\Leftarrow$)
Assume now that local diagnosers exist for $\tilde G$ satisfying bounded-delay detection and no false alarms.
By construction, fault occurrences are in one-to-one correspondence with membership in $K$, 
i.e., $s\in K$ iff $s\sigma_F\tau\in \mathcal{L}(\tilde G)$.  
Since $\sigma_F$ is unobservable, the diagnoser must base its decision entirely on observations of $s$ through $\Sigma_{o,i}$ (and $\tau$, which is non-informative after fault occurrence).
Thus the diagnoser induces local maps $S_1,S_2$ over the prefix $s$ such that
\[
F(S_1(P_1(s)),S_2(P_2(s)))=1 \iff s \in K.
\]
Therefore, these induced maps solve the original decentralized observation problem. 
\end{proof}

\subsection{Undecidability in Decentralized Fault Prognosis}
Another relevant problem widely investigated is  the decentralized  fault prognosis \cite{kumar2009decentralized,yin2018decentralized}. The basic setting is similar to the fault diagnosis problem. However, the objective of the \textbf{decentralized fault prognosis problem} is to decide the existence of local decision maps $S_i$ such that:
\begin{itemize}
\item 
\emph{No false alarm:}
whenever the fusion center issues an alarm, a fault is guaranteed to occur, i.e., 
for any $s \in \mathcal{L}(G)$, we have
\begin{align}
&    F(S_1(P_1(s)),S_2(P_2(s))) = 1\\
 \Rightarrow& 
(\exists n \in \mathbb{N}) (\forall t\in \mathcal L(G)/s)
\ [|t| \geq n\Rightarrow \sigma_F \in st].\nonumber
\end{align}
\item 
\emph{No missed fault:}
the fusion center must issue an alarm strictly before the fault occurrence, i.e., 
for any $ s \in \Psi(\sigma_F)$, we have
\begin{align}
(\exists s' \in \mathrm{Pref}(s)) 
[F(S_1(P_1(s')),S_2(P_2(s')))\!=\!1\wedge \sigma_F \notin s'].
\end{align}
\end{itemize}

Still, this decentralized fault prognosis problem is undecidable even when the local decision alphabet is finite. Since its proof is very similar to that of the decentralized supervisory control problem, we only sketch the proof briefly. 
\begin{mythm}
Let \((G,\Sigma_{o,1},\Sigma_{o,2},K,\Delta,F)\) be an instance of the decentralized observation problem. Then there exists a polynomial-time reduction from this instance to an instance of the decentralized fault prognosis problem. Consequently, the decentralized fault prognosis problem is undecidable, even under binary decision alphabets and XOR fusion.
\end{mythm}
\begin{proof}[Proof sketch]
The reduction is analogous to that used for decentralized supervisory control. The key idea is to encode the specification language $K$ as the set of states at which the fault event is defined, i.e., the boundary region.
We can construct a polynomial-time transformation that maps states in $Q_K$ to boundary (predictable-fault) states, and states in $\mathcal{L}(G)\setminus K$ to  states, at which non-fault loops are defined. Under this encoding, issuing a prognostic alarm corresponds exactly to entering the boundary region, which plays the same role as the decision-triggering states in the decentralized observation problem.
\end{proof}

\section{Conclusion}\label{sec:conclusion}

In this work, we show that decentralized decision-making problems for finite-state systems are undecidable, even under the highly restricted setting where each local agent communicates a single binary decision symbol to a central fusion station.
In particular, we identify the exclusive-or (XOR) fusion rule as a fundamental source of undecidability. Unlike classical conjunctive or disjunctive architectures, XOR is inherently non-monotone and lacks a natural control-theoretic interpretation, which prevents the application of standard decidable techniques used in decentralized supervisory control.
Our results reveal a fundamental limitation of decentralized decision-making architectures: the choice of fusion rule does not merely affect computational efficiency or expressive power, but can fundamentally alter the decidability of the local decision map synthesis problem itself.

These findings suggest several directions for future research. First, it remains open to characterize structural conditions on fusion rules, such as forms of logical monotonicity, under which decentralized decision-making becomes decidable. Such a characterization may provide a unified view of existing decidable architectures. Second, although the problem is undecidable in general, practical applications often impose additional restrictions, such as bounded memory at local agents. In such cases, local decision maps can be implemented as finite-state machines with bounded complexity, making the synthesis problem decidable via exhaustive enumeration in principle,  despite limited scalability.
Therefore, it is of interest to investigate bounded-synthesis techniques for constructing implementable local decision maps under memory constraints.

\bibliographystyle{plain}
\bibliography{references}

\end{document}